\documentclass[nonacm,sigconf]{acmart}
\usepackage{graphicx} %
\usepackage{graphicx}  %
\usepackage{natbib}  %
\usepackage[ruled,noend,linesnumbered]{algorithm2e}
\usepackage{amsmath}
\usepackage{amsthm}
\usepackage{todonotes}
\usepackage{macros}
\usepackage{caption}  %
\DeclareCaptionStyle{ruled}%
{labelfont=normalfont,labelsep=colon,strut=off}
\title{On the Complexity of Inductively Learning Guarded Clauses}

\author{Andrei Draghici}
\affiliation{%
  \institution{University of Oxford}
  \country{United Kingdom}
}
\email{andrei.draghici@stcatz.ox.ac.uk}

\author{Georg Gottlob}
\affiliation{%
  \institution{University of Oxford}
  \country{United Kingdom}
}
\email{georg.gottlob@cs.ox.ac.uk}

\author{Matthias Lanzinger}
\affiliation{%
  \institution{University of Oxford}
  \country{United Kingdom}
}
\email{matthias.lanzinger@cs.ox.ac.uk}

\begin{document}
\begin{abstract}
  We investigate the computational complexity of mining guarded
  clauses from clausal datasets through the framework of inductive
  logic programming (ILP). We show that learning guarded clauses is
  \np-complete and thus one step below the \sigmatwo-complete task of
  learning Horn clauses on the polynomial hierarchy. Motivated by
  practical applications on large datasets we identify a natural
  tractable fragment of the problem. Finally, we also generalise all
  of our results to $k$-guarded clauses for constant $k$.
\end{abstract}
\maketitle

\section{Introduction}
\label{sec:intro}

Finding complex relationships in large datasets with current data
analysis methods primarily relies on the computation of statistical
measures or simple aggregations over numerical data. Especially when
relationships between some parts of the data are conditional on other
data elements, standard methods such as correlation analysis reach
their limits. To better identify conditional relationships, one can alternatively approach the problem by trying to learn
logical relationships in the data, i.e., searching for non-trivial logical formulas
that hold true over the data set.

Learned logical formulas should ideally be intuitively interpretable
as well as usable in automated reasoning.  A natural fit for these
requirements are logical rules (i.e., Horn clauses) as used in Datalog~\cite{DBLP:books/cs/Ullman89} and the associated task of finding rules that
are satisfied in a dataset is often called \emph{rule mining}~\cite{DBLP:series/cogtech/FurnkranzGL12}. Not
only are modern rule-based languages highly expressive and theoretically well understood, but there also
exist real-world systems that are capable of efficient reasoning over
large amounts of data, e.g.~\cite{DBLP:journals/pvldb/BellomariniSG18}. In combination with such reasoning systems, learned rules can also be
used for data synthesis, anomaly detection, and similar tasks. 

This paper investigates the underlying theoretical problems when
viewing the rule mining problem through the lens of inductive logic
programming (ILP). Recall, in the ILP setting we are given a set of
positive examples and a set of negative examples and wish to find some
logical formula (the hypothesis) that subsumes all of the positive
examples while not subsuming any of the negative examples. We are
particularly interested in the computational complexity of the problem
and therefore consider the associated decision problem, also referred
to as the \emph{ILP consistency problem}, whether any solution exists
for the given sets of examples. 

The complexity of the ILP consistency problem depends critically on
the logical languages that are used for examples and the
hypothesis. In line with our motivation of learning rules, particular
historical attention has been given to the case where examples and
hypotheses are Horn clauses, for which the ILP consistency problem was shown to be
\sigmatwo-complete~\cite{DBLP:conf/ilp/GottlobLS97}.
While some problems in \np can now be solved efficiently in practical
settings, the second level of the polynomial hierarchy is still
generally out of reach for even moderately sized problems.  In the
context of mining rules from large datasets, \sigmatwo-completeness is
therefore clearly prohibitive and limits potential applications to
only small and simple scenarios.

To address this challenge we consider the learning of \emph{guarded
  clauses}, i.e., we restrict the language of the hypothesis to only
guarded clauses. Informally, a clause is said to be \emph{guarded} if
there exists one literal in the clause that \emph{covers} all
variables that occur in the clause. Our interest in guarded clauses is
motivated from two sides. First, as discussed below, we see that
learning guarded clauses can reduce the complexity of the
problem. Second, in common rule-based languages, e.g., Datalog or
Datalog$^\pm$~\cite{DBLP:conf/pods/CaliGL09}, programs consisting only
of guarded rules also enjoy various desirable computational
properties~\cite{DBLP:journals/tocl/GottlobGV02,
  DBLP:conf/pods/CaliGL09}\footnote{Guarded rules in the context of
  Datalog$^\pm$ and guarded clauses have slight technical
  differences. These differences are
  inconsequential to our results which all also apply to learning
  guarded rules in way the term is used in Datalog$^\pm$, see the discussion in Section~\ref{subsec:alg}.}. While guarded rules are of particular interest in practice, we present our results for the more general problem of learning guarded clauses, i.e., disjunctive rules. However, all results hold also for Horn clauses. 

Importantly, guarded clauses and corresponding rule-based languages are still highly expressive. The guarded fragment of Datalog$^\pm$ captures popular logical languages, such as the description logics DL-Lite$_\mathcal{R}$ and $\mathcal{EL}$ which form the theoretical foundation of the OWL 2 QL and EL profiles, respectively~\cite{DBLP:journals/ws/CaliGL12}. In addition to the expressiveness, reasoning over guarded clauses can be more efficient than over general clauses~\cite{DBLP:journals/ws/CaliGL12}. The following example aims to further clarify our setting.
\begin{example}
  \label{ex1}
  Consider the following clauses, which make up a set of positive examples $\{E^+_1,E^+_2\}$ and the single negative example $E^-_1$ (with constants $a,b,c,d,e$).
\[
  \begin{array}{l}
    E^+_1:  \neg \mathsf{TalkAbout}(\mathit{a}, \mathit{b}, \mathit{a}) \lor
             \neg \mathsf{FanOf}(\mathit{a}, \mathit{a}) \lor\\ \qquad\;
             \neg \mathsf{Influences}(\mathit{a}, \mathit{b}) \lor
             \mathsf{FanOf}(\mathit{b}, \mathit{a}),
\\ \\
    E^+_2:    \neg \mathsf{TalkAbout}(\mathit{a}, \mathit{c}, \mathit{d}) \lor 
             \neg \mathsf{FanOf}(\mathit{a}, \mathit{d}) \lor \\ \qquad\;
             \neg \mathsf{Influences}(\mathit{a}, \mathit{c}) \lor
             \mathsf{FanOf}(\mathit{c}, \mathit{d}) \lor \; \mathsf{Parent}(\mathit{c}, \mathit{b})
    \\ \\
    E^-_1: \neg \mathsf{TalkAbout}(\mathit{d}, \mathit{b}, \mathit{e}) \lor
             \neg \mathsf{Influence}(\mathit{d}, \mathit{b}) \lor \\ \qquad\;
             \mathsf{FanOf}(\mathit{d}, \mathit{e})
  \end{array}
\]
The following guarded clause (written as a rule) is a natural ILP-hypothesis that can be learned from the above examples.
\[
  \begin{array}{ll}
    \mathsf{FanOf}(y,z) \leftarrow & \mathsf{TalkAbout}(x, y, z), \mathsf{FanOf}(x,z),\\ &  \mathsf{Influences}(x,y)
  \end{array}
\]
Intuitively, the rule can be interpreted as: if $x$ and $y$ talk about $z$, $x$ is a fan of $z$, and $x$ influences $y$, then $y$ will also be a fan of $z$.
Here, the atom $\mathsf{TalkAbout}(x, y, t)$ contains all three variables that are used in the clause, and therefore \emph{guards} the clause.
Alternatively, if we drop the guardedness requirement, the following clause would also be a valid hypothesis
\[
  \neg \mathsf{Influences}(x,y) \lor \mathsf{FanOf}(y, t)
\]
 Note that is difficult to interpret this clause intuitively, especially when viewed at as a rule: if $x$ influences $y$, then $y$ is a fan of $t$.
\end{example}

\paragraph{Related Work}%
Inductive logic programming and learning clauses in particular is a
classic theme of AI research and efficient algorithms for learning clauses have been the topic of a number of classic works, e.g.~\cite{DBLP:journals/ngc/Muggleton91,DBLP:conf/alt/MuggletonF90,DBLP:conf/icml/KietzL94}. 
The complexity of learning clauses has been studied in many other contexts and under a broad range syntactical restrictions. One particularly important case is the learning of Horn clauses. There, the ILP consistency problem is
$\sigmatwo$-hard~\cite{DBLP:conf/ilp/GottlobLS97}, even without the
presence of background knowledge. Despite the large body of work on learning clauses, to the best of our knowledge, guardedness has not yet been studied in this context.

Recently, the inductive synthesis of Datalog programs, using a variety
of different techniques, has received particular
attention~\cite{DBLP:conf/ijcai/SiRHN19,DBLP:journals/pacmpl/RaghothamanMZNS20}. The
resulting systems are able to synthesise complex logic programs,
including advanced use of recursion, background knowledge and
predicate invention.

However, our
motivation of rule mining in large datasets is not fully aligned with
this direction of research. In our setting we prefer learning many
individual rules instead of complex programs while also requiring
stronger upper bounds on the computational complexity to deal with
large inputs. Our results may be applicable to the synthesis of \emph{guarded} Datalog programs when considered in conjunction with some of the methods noted above.
For further details on new developments in the synthesis of logic programs and various other directions of current ILP research, we defer to a recent survey article~\cite{DBLP:conf/ijcai/CropperDM20}.

\paragraph{Our Contributions}
Our goal is to study the computational
complexity of learning guarded rules and find expressive but tractable fragments that can support real-world rule mining. We show that in the general
case, the problem is no longer \sigmatwo-hard but still
intractable. However, we are able to identify the source of the
intractability and demonstrate that an important fragment of the problem is
indeed tractable. The full contributions of this paper can be
summarised as follows:
\begin{enumerate}
\item First we show that guardedness in general improves the
  complexity in versus the restriction to Horn clauses. In particular,
  the problem moves one step down on the polynomial hierarchy from
  \sigmatwo to \np.
\item We show not only \np-completeness of the general problem but
  also for various common further restrictions such as bounded arity
  of relations and bounded number of variables in the hypothesis.
\item In response, we identify the notion of \straight clauses. We
  show that learning guarded clauses from a clausal dataset is
  tractable if the positive examples are straight. We present a
  polynomial time algorithm and furthermore, show that our algorithm
  also allows tractable enumerate of an class of representative
  solutions of the problem.
\item Finally, we argue that our results also apply analogously to the
  much more general case of guardedness by not one literal, but a
  constant number of literals (so-called $k$-guardedness).
\end{enumerate}

\paragraph{Structure} The rest of this paper is structured as
follows. We first recall important technical preliminaries in
Section~\ref{sec:prelims}. We demonstrate the \np-completeness of
learning guarded clauses in Section~\ref{sec:hardness}. Our tractable
algorithm, related theoretical results for learning from \straight
clauses, and the generalisation to $k$-guardedness is presented in
~\ref{sec:alg}. Concluding remarks and a brief discussion of possible
future work are given in Section~\ref{sec:conclusion}.

\section{Preliminaries}
\label{sec:prelims}

We assume that the reader is familiar with typical notions of
computational complexity as defined
in~\cite{DBLP:books/daglib/0018514}, in particular polynomial time
reductions and the first two steps of the polynomial hierarchy.

We begin by recalling some important standard notions of inductive logic
programming. We roughly follow the presentation
from~\cite{DBLP:journals/sigart/KietzD94}. We assume basic familiarity
with standard terminology of first-order logic as used, e.g.,
in~\cite{DBLP:books/daglib/0082098}.  In general, for a literal $L$ we will use
$\vars(L)$ to denote the set of variables occurring in the
atom. Likewise, we write $\termsf(L)$ for the set of all terms occurring
in $L$, and $\ar(L)$ for the arity of $L$.

Let $\mathcal{L}$ be the language of first order logic. An \emph{ILP
  setting} is a tuple
$\left<\mathcal{L}_B, \mathcal{L}_H, \vdash, \mathcal{L}_{E}\right>$
where $\mathcal{L}_B, \mathcal{L}_H, \mathcal{L}_E$ are subsets of
$\mathcal{L}$ and $\vdash$ is a correct provability relation.

The \emph{ILP consistency problem} for a ILP setting $\mathcal{S}$
denotes the following decision problem: Given sets
$B, \expos \cup \exneg$ where $B \subseteq \mathcal{L}_B$, and
$\expos,\exneg \subseteq \mathcal{L}_E$, does there exist a
\emph{hypothesis} $H \in \mathcal{L}_H$ such that
$B \cup \{H\} \vdash \expos$, $B \cup \{ H\} \not \vdash \exneg$, and
$B \cup \{H\} \not \vdash \bot$ ($B \cup H$ is consistent). In such an
instance, $B$ is commonly referred to as the \emph{background
  knowledge}, while \expos, \exneg are the positive and negative
\emph{examples}, respectively. Furthermore, we will refer to any such
hypothesis as a \emph{solution} to the respective instance.

Note that for a set of formulas $E$, $T \vdash E$ is short for
$T \vdash e$ for every $e \in E$, and $T \not \vdash E$ is short for
$T \not\vdash e$ for every $e \in E$. Hence, a valid hypothesis must
subsume all positive examples and none of the negative examples.

By slight abuse of notation we will also consider ILP settings where
the languages for the examples $\mathcal{L}_E$ is split into separate
languages for the positive and the negative examples. Thus $\expos$
above will consist of clauses from a language $\mathcal{L}_{E^+}$,
while $\exneg$ are clauses from a possibly different language
$\mathcal{L}_{E^-}$. The languages can of course still overlap and
this variation of ILP settings has no effect on other definitions.

As is common in the field, we will consider specifically languages of
first order clauses. It will be convenient to treat a clause as a set
of the individual disjuncts and we will do so implicitly throughout
this paper. In this spirit we also assume no duplicate literals in clauses. It will be useful to distinguish between a relation name occurring in a positive and a negative literal. We use the term \emph{signed relation name} to refer to a relation name together with the polarity of the literal in which it occurs.
Following common practice in the field, we focus on
\emph{$\theta$-subsumption} $\vdash_\theta$ as a provability
relation here. Recall, we say that a clause $\phi$ $\theta$-subsumes a
clause $\psi$, i.e., $\phi \vdash_\theta \psi$, if there exists a
variable substitution $\theta$ such that $\phi \theta \subseteq \psi$.
As we will always use $\vdash_\theta$ as the provability relation, we
will generally drop the subscript and use $\vdash$ to implicitly
represent $\theta$-subsumption.

We are now ready to recall the important result for Horn clauses
mentioned in the introduction. Let $\mathcal{S}_{HC}$ be the ILP
setting
$\left<\emptyset, \text{function-free Horn clauses}, \vdash_\theta,
  \text{ground clauses}\right>$.
\begin{proposition}[\cite{DBLP:conf/ilp/GottlobLS97}]
  The ILP consistency problem for $\mathcal{S}_{HC}$ is \sigmatwo-hard.
\end{proposition}

For the rest of this paper we are interested in the problem of
learning \emph{guarded} clauses.  We say that a clause $C$ is
\emph{guarded} if there is some literal $G \in C$ such that
$\bigcup_{R \in C} \vars(R) \subseteq \vars(G)$. Guardedness can also
be generalised to so-called \emph{$k$-guarded} clauses, where some
group of $k$ literals in the clause covers all variables of the
clause.  Much of this paper will focus on the ILP setting
$\mathcal{G} = \left<\emptyset, \text{function free guarded clauses},
  \vdash_\theta, \text{ground clauses}\right>$.
Since $\mathcal{G}$ does not allow for background knowledge, we will treat
instances for the ILP consistency problem of $\mathcal{G}$ simply as 
tuples of positive and negative examples. Some notes on the inclusion of background knowledge in our
setting are made in Section~\ref{sec:conclusion}.

\section{\np-Completeness of Guarded ILP Consistency}
\label{sec:hardness}

Note that the language of guarded clauses is incomparable to the language of Horn clauses. Our setting is therefore not simply a more restricted version of previously studied work but an alternative restriction, which we will show to have more desirable computational properties.

First, in this section, we show that the ILP consistency problem for $\mathcal{G}$ is \np-complete. This represents a significant improvement over the complexity of $\mathcal{S}_{HC}$. The proof details will allow us to identify a practically useful tractable fragment of the problem, which we present in next section.

We will first show a reduction from  the hitting string problem~\cite{DBLP:books/fm/GareyJ79} to establish \np-hardness. Afterwards we present an argument for \np membership. Note that membership of the problem is not straightforward. It is not trivial whether it is always sufficient to guess a polynomially bounded hypothesis. Furthermore, it is not easy to verify a guess in polynomial time since checking $\theta$-subsumption is \np-hard in general~\cite{DBLP:conf/icml/KietzL94,DBLP:books/fm/GareyJ79}.

\subsection{Reduction from Hitting String}
In the \emph{hitting string} problem we are given a finite set of
strings $S = \{ \mathbf{s_1}, \dots, \mathbf{s_m}\}$, all of length
$n$, over the alphabet $\{0, 1, *\}$. The task is to decide whether
there exists a \emph{binary} (over the alphabet $\{0,1\}$) string
$\mathbf{x} = x_1x_2\dots x_n$ such that for each
$\mathbf{s_i} \in S$, $\mathbf{s_i}$ and $\mathbf{x}$ agree in at
least one position. Such a string $\mathbf{x}$ is called a
\emph{hitting string} for the set $S$. The hitting string problem is
known to be \np-complete~\cite{DBLP:books/fm/GareyJ79}.

\begin{theorem}
  \label{thm:nphard}
  The ILP consistency problem for $\mathcal{G}$ is \np-hard. The problem remains \np-hard even when hypothesis and example languages are additionally restricted to clauses with any combination of the following properties:
  \begin{itemize}
  \item Horn,
  \item arity at most 2,
  \item at most 2 variables in the hypothesis,
  \item or no repetition of terms in a single literal.
  \end{itemize}
\end{theorem}
\begin{proof}
  Proof is by reduction from the hitting string problem. Let
  $S= \{ \mathbf{s_1}, \dots, \mathbf{s_m}\}$ be set of length $n$
  strings over the alphabet $\{0,1,*\}$. Let the positive examples
  $\expos$ be the set $\{ C_i \mid 0 \leq i \leq n \}$ with the
  respective clauses defined as
\[
  \begin{array}{ll}
    C_0 = & G(a,b) \lor \bigvee_{1\leq j \leq n} \left( A_j(a) \lor B_j(b) \right) \\[1.2ex]
    C_i = & G(a,b) \lor G(b,a) \lor A_i(a) \lor B_i(a) \lor \\
          & \bigvee_{1 \leq j \leq n, j \neq i} \left( A_j(a) \lor A_j(b) \lor B_j(a) \lor B_j(b) \right)
  \end{array}
\]
For string $\mathbf{s_i} \in S$, we write $s_{i,j}$ to refer to the symbol at position $j$ of the string. 
Let set of negative examples $\exneg$ be the set $\{D_i \mid 0 \leq i \leq n\} \cup \{ N_i \mid 0 \leq i \leq m\}$ where the individual clauses are as follows
\[
  \begin{array}{ll}
    D_0 = & \bigvee_{1\leq j \leq n} \left( A_j(a) \lor A_j(b) \lor B_j(a) \lor B_j(b)\right) \\[1.2ex]
    D_i = & G(a,b) \lor \bigvee_{1 \leq j \leq n, j \neq i} \left( A_j(a) \lor B_j(b) \right) \\[1.2ex]
    N_i = & G(a,b) \lor \bigvee_{j: s_{i,j}=1} B_j(b) \lor \bigvee_{j: s_{i,j}=0} A_j(a)  \\
          & \lor \bigvee_{j: s_{i,j}=*}\left(  A_j(a) \lor B_j(b) \right)
  \end{array}
\]

Now suppose there is a function-free guarded clause $H$ such that $H \vdash \expos$ and $H \vdash \exneg$. We first make some observations on the structure of $H$ and then argue how this connects to the original hitting string instance.
First, observe that any hypothesis $H$ must contain some atom for relation $G$ to not subsume $D_0$. On the other hand, by $C_0$, we see that both positions must contain different variables and only one of $G(x,y)$ and $G(y,x)$ is in  $H$. Hence, w.l.o.g., we assume only $G(x,y)$ is in $H$.

Since $H$ must subsume $C_0$, there can be no atoms $A_i(y)$ or $B_i(x)$ in $H$. In consequence, since $D_i$ may not be subsumed, at least one of $A_i(x)$ or $B_i(y)$ must occur in $H$. Indeed, because $H$ needs to subsume $C_i$ it follows that only exactly one of $A_i(x)$ or $B_i(y)$ is in $H$ for every $1 \leq i \leq n$. Hence, we know that $H$ is always of the form
\[
  G(x,y) \lor \bigvee_{1 \leq i \leq n} \Gamma_i,\quad \text{where } \Gamma_i \text{ is either } A_i(x) \text{ or } B_i(y)
\]
Note that $H$ is clearly always guarded.
We can associate a binary string $\mathbf{x}(H)$ of length $n$ to each clause $H$ of this form by simply setting position $i$ of the string to be $1$ if $A_i(x) \in H$, or $0$ if $B_i(y) \in H$. 

We are now ready to show the correctness of the reduction, in particular we argue that $\mathbf{x}(H)$ is a hitting string for $S$ if and only if $H \vdash \expos$ and $H \not\vdash \exneg$.

Assume $\mathbf{x}$ is a hitting string for $S$ and let $H$ be the clause associated to $\mathbf{x}$. From our arguments above we know that the only thing left to check is whether $H$ does not subsume any of the $N$ clauses. Every string $\mathbf{s}_i\in S$ agrees with $\mathbf{x}$ on some position, i.e., $\mathbf{s}_{i,j} = \mathbf{x}_j$. If such a $\mathbf{s}_{i,j}$ is $1$, then $A_j(x) \in H$ but $A_j(a) \not\in N_i$. The analogous reasoning holds for the case where the position contains $0$. Since $\mathbf{x}$ is a binary string, these are the only two options and we see that for every string $\mathbf{s}_i \in S$, the clause $N_i$ is not subsumed by $H$. Hence, $H \vdash \expos$ and $H \not \vdash\exneg$.

For the other side of the implication, suppose there is a guarded clause $H$ with $H \vdash \expos$, $H \not \vdash \exneg$. From our arguments before we know the structure of $H$ and that the clause determines a unique binary string $\mathbf{x}(H)$. Now, for every $0 \leq i \leq m$ we have that $H \not \vdash N_i$. By inspection of $H$ and $N_i$ it follows that for some $0 \leq j \leq n$, either $A_j(x) \in H$ and $\mathbf{s}_{i,j}=1$, or $B_j(y) \in H$ and $\mathbf{s}_{i,j}=0$. In either case, by definition of $\mathbf{x}(H)$ we see that the $j$-th position of $\mathbf{x}(H)$ agrees with the value of $\mathbf{s}_{i,j}$. Since this holds for all positions, it follows that $\mathbf{x}(H)$ agrees with every string in $S$ at some position, i.e., $\mathbf{x}(H)$ is a hitting string for $S$.

This establishes the \np-hardness for $\mathcal{G}$. By inspection of the clauses in the reduction, we can deduce that the problem remains hard under the stated additional restrictions
All of the clauses are positive, i.e., polarity plays no role in the argument. It is also easy to adapt each clause to be Horn instead
(make all atoms in the above reduction negative, add a new redundant literal $P(a)$ to every clause).
Lastly, only arity 2 atoms, 2 variables in $H$ and 2 constants in the examples are necessary.
\end{proof}

Note that our reduction also demonstrates that the restriction to guardedness is not itself a source of difficulty. The above argument does not require the restriction to guarded hypothesis. It follows from our observations on the form of solutions to our reduction that any clause that subsumes all of $\expos$ and none of $\exneg$ must be guarded.

\subsection{\np Membership}
We first establish that it is sufficient to guess a polynomially bounded hypothesis. Note that this is not true in every setting and previous work has instead regularly considered the \emph{bounded ILP consistency problem}~\cite{DBLP:conf/pkdd/AlphonseO09,DBLP:conf/ilp/GottlobLS97}, where explicit constraints on the hypothesis size are considered as part of the decision problem. 

While guardedness still allows for exponentially large solutions, most literals in such a solution will be  redundant. In particular, we can observe that if a guarded hypothesis does not subsume a clause, then there are only a few concrete variable assignments for which subsumption must fail in addition to those cases  where the guard itself can not be  subsumed. A solution then needs only one literal per such case to guarantee non-subsumption of the clause. We can therefore observe the following bound (a full argument is given in the technical appendix).
\begin{lemma}
    \label{lem:smallH}
  Assume the instance $\expos, \exneg$ for the ILP consistency problem for $\mathcal{G}$ has a solution. Then it also has a solution $H$ with $|H| \leq 1+\sum_{E \in \exneg} |E|$.
\end{lemma}

The natural guess and check procedure for the ILP consistency problem
is to guess a hypothesis and then check (non-)subsumption for each of
the examples. However, in general this does not result in an $\np$
algorithm since deciding whether $\phi \vdash \psi$ is \np-complete
even for function-free Horn clauses~\cite{DBLP:conf/icml/KietzL94}.

In the guarded setting, subsumption can be decided more efficiently.
It is not difficult to prove directly that subsumption by a
function-free guarded Horn clause is tractable. Due to our later
interest in $k$-guardedness it will be more convenient to adapt a more
general result on the complexity of finding homomorphisms.

It is straightforward to translate the problem $\phi \vdash \psi$,
where $\phi$ is a function-free clause, to a the question of whether
there exists a homomorphism from a relational structure representing
$\phi$ to a structure representing $\psi$. The computational complexity of homomorphism checking is well understood~\cite{DBLP:journals/jacm/Grohe07,DBLP:journals/jcss/GottlobLS02}, and by applying this knowledge to our setting we can make the following important observation.

\begin{lemma}
 \label{prop:kguardsubsumption}
 Let $k$ be an integer constant, let $\phi$ and $\psi$ be clauses such
 that $\phi$ is $k$-guarded. Then it is tractable to decide whether
 $\phi \vdash_\theta \psi$.
\end{lemma}

Naturally, the complementary problem $\phi \not \vdash \psi$ is tractable under the same circumstances. Since we know from Lemma~\ref{lem:smallH}, that it is sufficient to guess a polynomially bounded $H$, it is clear that the ILP consistency problem for $\mathcal{G}$ is in \np. In combination with Theorem~\ref{thm:nphard} we arrive at the main result of this section.

\begin{theorem}
  The ILP consistency problem for $\mathcal{G}$ is \np-complete.
\end{theorem}

\section{Tractability for Straight Positive Examples}
\label{sec:alg}

In the previous section we have seen that considering guarded clauses
instead of Horn clauses improves the complexity of the consistency problem form $\sigmatwo$-completeness to $\np$-completeness. However, in the context of our original motivation
of rule mining on large datasets this may still be infeasible. Furthermore,
Theorem~\ref{thm:nphard} illustrates that even severe restrictions to
the clause languages are insufficient to achieve a polynomial time
algorithm (assuming $\np \neq \ptime$).

In this section we introduce the notion of \straight
clauses which will allow us to define a tractable fragment of the problem that can still express many common notions, such as recursive rules. We then prove some important theoretical properties under this new restriction and ultimately give a polynomial time algorithm. Finally, we briefly discuss how key results and our algorithm generalise to $k$-guardedness.

\begin{definition}{Straight Clauses}
  A clause $\phi$ is called \emph{\straight} if no relation symbol occurs twice in $\phi$ with the same polarity.
\end{definition}

Note that \straight clauses can still express recursive rules since the same relation can appear twice with differing polarity. For example, all clauses and the resulting hypothesis in Example~\ref{ex1} are \straight clauses. In the context of Datalog rules, \straight rules correspond to rules with no self-joins.

It will not be necessary to restrict all clauses to be \straight. Indeed, we will only require positive examples to be \straight clauses, whereas the negative clauses and the hypothesis will not need to be \straight. To that end we define the new ILP setting $\mathcal{G}^*$ as the setting $\mathcal{G}$ where the positive examples additionally have to be \straight clauses.

\subsection{Guarded Hypotheses for Straight Clauses}
\label{sec:det}

We first introduce some new notation. In a clause $C$ where a signed relational name $P$ occurs only once we will use $\pi^C_{P,i}$ to denote the term at the $i$-th position of the literal corresponding to $P$.
\begin{definition}[Relative Shields]
  Let $Q$ and $P$ be relational symbols occurring in a \straight clause $C$. The \emph{relative shields} of position $i$ of $Q$ are defined as the set 
  \[
  S_{P,Q,i}^C = \{ j \mid \pi^C_{Q,j} = \pi^C_{P,j} \}
  \]
  That is, the relative shields are all the positions of $P$ that contain the exact same term as position $i$ of $Q$.
  
  Let $E$ be a set of \straight clauses that all contain $P$ and $Q$. The relative shield of position $i$ of $Q$ by $P$ (w.r.t. $E$) is the set 
  \[
  S^E_{P,Q,i} = \bigcap_{C\in E} S_{P,Q,i}^C
  \]
\end{definition}
\begin{example}
  Consider the following two clauses
  \[
    \begin{array}{l}
      P(a,a,b,a),\, Q(a)  \\
      P(a,b,a,b), \, Q(b) 
    \end{array}
  \]
  The relative shield of position 1 of $Q$ by $P$ is $\{1,2,4\}$ in the first clause and $\{2,4\}$ in the second clause. Thus, for the set of clauses we have the relative shield $\{2,4\}$. In situations where $P$ is the guard of the hypothesis $H$ the relative shield corresponds to the variables from the guard clause in $H$ that can be used in the respective position of $Q$. That is, if $P(x,y,x,z)$ is the guard in $H$, then only $Q(y)$ or $Q(z)$ are possible candidates for the hypothesis.
\end{example}

The following lemma formalises the observation made at the end of the example. It shows how the relative shield and variable choice are tightly connected in guarded clauses.
\begin{lemma}
\label{lem:shielding}
    Let $E$ be a set of \straight clauses and let $H$ be a clause with guard $P$ such that $H \vdash E$. For every literal $Q(x_1,\dots,x_n) \in H$, we have that for every $1 \leq i \leq \operatorname{ar}(Q)$, if $x_i = \pi^H_{P,j}$, then $j \in S_{P,Q,i}$. (The relative shield is again considered with respect to $E$).
\end{lemma}

We will see that beyond Lemma~\ref{lem:shielding}, the relative shield will play a special role if we consider a particular guard, namely the the least general one.

\begin{definition}[Least General Induced Guard]
  Let $E$ be a set of \straight ground clauses that all contain the signed relation name $P$. The least general induced guard (lgig) for $P$ (induced by $E$) is the literal $P(\bar{x})$ with the least number of distinct variables such that $P(\bar{x}) \vdash \{P(\bar{a})\}$, for all $P(\bar{a})$ in any clause in $E$.
\end{definition}

It is easy to observe that the lgig for $P$ has the same variables exactly in those positions that contain the same term in every clause. That is, $x_i=x_j$ if and only if $\forall C \in E\colon \pi^C_{P,i} = \pi^C_{P,j}$, where $x_k$ refers to the variable in position $k$ of the lgig. It follows that the lgig is unique up to isomorphism.

In the next step, We will show how using an lgig as guard can determine the exact argument list of all literals in the hypothesis. This is an important observation for our eventual algorithm. Furthermore, it will allow us to show that it is always sufficient to consider only lgigs as guards in our scenario.

\begin{lemma}
  \label{lem:uniqshield}
  Let $E$ be a set of \straight ground clauses and let $H$ be a clause guarded by the lgig for some signed relation name $P$ such that $H \vdash E$. For every literal $Q(\dots) \in H$, we have that for every $i \in \mathit{ar}(Q)$ and every $j,k \in S^E_{P,Q,i}$ we have $\pi^H_{P,j}=\pi^H_{P,k}$.
\end{lemma}
\begin{proof}
  Suppose the statement fails for $Q(x_1,\dots,x_n) \in H$ at position $i$. That is, there are $j,k \in S^E_{P,Q,i}$ such that there are different variables at positions $j$ and $k$ in the guard $P$, i.e., $\pi^H_{P,j}\neq \pi^H_{P,k}$.
  Observe that if $j, k$ are both relative shields, then it holds that $\pi^C_{P,j} = \pi^C_{Q,i} = \pi^C_{P,k}$ for all $C \in E$. Therefore, for every $C\in E$ there exists a $\theta$ such that $H\theta \subseteq C$ where $\theta$ maps $\pi^C_{P,j}$  and $\pi^C_{P,k}$ to the same term. Thus, let $\theta'$ be the substitution that maps $\pi^C_{P,j} \mapsto \pi^C_{P,k}$ and everything else to identity. Clearly, $H\theta' \vdash E$ and the guard in $H\theta'$ has symbol $P$ and fewer distinct variables than the guard in $H$, hence contradicting the assumption of the guard in $H$ being the lgig.
\end{proof}

Since the above lemma uniquely determines the argument list of each literal occurring in an lgig guarded hypothesis, it also follows that no signed relation name can occur twice in the hypothesis. Hence, the hypothesis is also always \straight.

Another important property of the lgig is related to the following notion of specialisation.
We call a variable substitution $\theta$ a \emph{specialisation} if it maps from the variables $X$ to a subset $X' \subset X$ such that $\theta(x)=x$ for all $x \in X'$. The following lemma establishes not only that if there is a solution, there is also a solution guarded by the lgig, but that any guard in a solution can be specialised to be an lgig.

\begin{lemma}
\label{lem:specialize}
    Let $E$ be a set of \straight grounded clauses and let $H$ be a clause with guard literal $P(\bar{x})$ such that $H \vdash E$. There exists a specialisation $\theta$ such that $P(\bar{x})\theta$ is the lgig of $P$ induced by $E$.
\end{lemma}
\begin{proof}
Let $P(\bar{y})$ be the lgig of $P$ induced by $E$.
    Recall that $y_i=y_j$ if and only if $\forall C \in E\colon \pi^C_{P,i} = \pi^C_{P,j}$. Observe that if two positions $i,j$ of $P$ differ in some clause in $E$ (and if $H \vdash E$), then the variables at the same positions of $P(\bar{x})$ in $H$ can not be the same. Otherwise the same variable would have to map to different terms to match the single occurrence of $P$ in the clause.
    
    With this established, we see that in $P(\bar{x})$ we have $x_i = x_j$ only if $y_i = y_j$ in the lgig $P(\bar{y})$. Thus, clearly the substitution $\theta: x_i \mapsto y_i$ for all $1 \leq 1 \leq \operatorname{ar}(P)$ is a specialisation and $P(\bar{x})\theta$ is the lgig of $P$ induced by $E$.
\end{proof}
From these key lemmas we can derive the following key theorem for the setting $\mathcal{G}^*$. A full proof of the missing details is given in the technical appendix.

\begin{theorem}
  \label{thm:key}
  Let $\expos$ be a set of \straight ground clauses, and let $\exneg$ be a set
  of ground clauses. If there exists a guarded clause $H$ such that
  $H\vdash \expos$ and $H \not\vdash \exneg$, then there also exists a
  guarded clause $H'$ with the following properties:
    \begin{enumerate}
        \item The guard of $H'$ is a least general induced guard (induced by $\expos$).
        \item $H'$ is a \straight clause.
        \item $H' \vdash \expos$ and $H' \not \vdash \exneg$
    \end{enumerate}
\end{theorem}

\subsection{A Polynomial-Time Algorithm}
\label{subsec:alg}

The results of Section~\ref{sec:det} show that we can use least general induced guards to fix a mapping of terms in a positive example to variables in the hypothesis. By Lemma~\ref{lem:uniqshield}, this mapping, in combination with a ground literal from the example, determines the only possible way to include a corresponding literal (with the same polarity and relation symbol) in the hypothesis.

Lemma~\ref{lem:uniqshield} is a crucial observation for our algorithm. The algorithm iteratively adds literals from a positive example to the hypothesis by replacing the terms there with variables. In general, the choice of how exactly terms are replaced by variables leads to combinatorial explosion. However, Lemma~\ref{lem:uniqshield} states that we can limit our search to one particular choice if the hypothesis is guarded by an lgig.

\begin{definition}
  \label{def:map}
  Let $E$ be a set of \straight ground clauses and let $H$ be a clause guarded by the lgig for $P$.
  Let $L$ be a literal in some clause of $E$. The \emph{lgig map} $\mu$ of $L$ is a function $\termsf(L) \to \vars(H)$ such that $\mu(t_i) = \pi^H_{P,j}$ where $j \in S^E_{P,L,i}$ and $t_i$ is the term at position $i$ of $L$.
\end{definition}

It follows from Lemma~\ref{lem:uniqshield} that such a map always exists -- and is uniquely determined -- if a literal with the same name and polarity as $L$ is part of a solution guarded by the lgig for $P$.  Note that no lgig map exists if the relative shield is empty at any position $i$.

With the above definition, we are now ready to present Algorithm~\ref{alg}, our polynomial time algorithm for the ILP consistency problem for $\mathcal{G}^*$. To avoid redundant checks throughout the algorithm we assume, w.l.o.g., that the instance is non-trivial, i.e., that $\expos \neq \emptyset$ and that all signed relation names occur in all clauses of $\expos$. Clearly, any signed relation name that does not occur in all clauses can never be part of the hypothesis, as it would be impossible to subsume a clause that does not contain the relation.

The algorithm operates by iteratively building hypotheses. For some arbitrary clause $C \in \expos$, the algorithm will try every signed relation name in the clause as a potential guard (line~\ref{line:guard}). For each such candidate, the lgig induced by $\expos$ is generated and added to a candidate hypothesis $H$. From this point on the algorithm will try to add further literals from $C$ to $H$ via their lgig map as long as they do not break subsumption of $\expos$. This either continues until $H$ is identified as a solution (if also $H \not \vdash \exneg$), or until all possible additions are exhausted. If no solution is found for any choice of $G$, then the algorithm rejects.

\begin{algorithm}[t]
  \SetKwInOut{Input}{input}\SetKwInOut{Output}{output}
  \SetKw{Reject}{Reject}
  \SetKw{Continue}{Continue}

  \Input{Examples $\expos\neq \emptyset$, \exneg as sets of grounded straight clauses.}
  \Output{Guarded clause $H$ s.t. $H \vdash \expos$ and $H \not\vdash \exneg$}

  $C \leftarrow $ an arbitrary clause in $\expos$\;
  \For{$G \in C$}{\label{line:guard}
    $G' \leftarrow$ be the least general induced guard for $G$ induced by $\expos$\;
    $H \leftarrow \{ G' \}$\;
    \If{$H \not\vdash \exneg$}{
      \Return $H$\;
    }

    \For{$L \in C \setminus \{G\}$}{
      \If{an lgig map of $L$ w.r.t. $G'$ exists}{
        $\mu \leftarrow $ lgig map of $L$ w.r.t. lgig $G'$\;
        
        \If{$H \cup \{\mu(L)\} \vdash \expos$}{
          $H \leftarrow H \cup \{ \mu(L) \}$\;
        }
      }
    }
    \If{$H \not\vdash \exneg$}{
      \Return $H$\;
    }
  }
  \Reject{}
  \caption{A polynomial time algorithm for ILP consistency for setting $\mathcal{G}^*$.}
  \label{alg}
\end{algorithm}

It is straightforward to observe the soundness of Algorithm~\ref{alg}. By definition, any returned clause $H$ is a guarded clause with $H \vdash \expos$ and $H \not \vdash \exneg$. What is left to show is that the algorithm is complete, i.e., that it will always return a hypothesis if the instance has a solution, and that it is tractable. Completeness follows from the observations on lgig maps made above and Theorem~\ref{thm:key}, showing that if a solution exists, then a \straight solution guarded by a lgig exists. A full argument is given in the technical appendix.

\begin{theorem}
\label{lem:algcomplete}
Algorithm~\ref{alg} is a correct algorithm for the ILP consistency problem for $\mathcal{G}^*$.
\end{theorem}

The worst-case time complexity of Algorithm~\ref{alg} depends on the complexity of computing the lgig, the lgig maps, and checking subsumption. Computing the lgig for literal $P$ is effectively in $O(\ar(P)^2 \cdot |\expos|)$, since it only requires the identification of positions that contain equal terms in all clauses. Similarly, from Lemma~\ref{lem:uniqshield} we see that the lgig map for a literal $L$ can be constructed directly from the respective relative shields. Again this can be done in time $O(\ar(L)^2 \cdot |\expos|)$, followed by a lookup of the relevant positions in the guard literal.
As stated previously in Lemma~\ref{prop:kguardsubsumption}, checking subsumption is indeed tractable for guarded formulas. Clearly, the hypothesis is guarded in all subsumption checks in Algorithm~\ref{alg} and therefore all of the checks are in polynomial time (actually even linear and in $\nctwo$, see~\cite{DBLP:journals/tcs/GottlobLS02}) .

\begin{theorem}
  The ILP consistency problem for $\mathcal{G}^*$ is in $\ptime$.
\end{theorem}

\paragraph{Enumerating representative solutions}
While the decision problem is important for theoretical considerations, in practice the enumeration of solutions is of greater interest. In our setting $\mathcal{G}^*$, it is still possible to have exponentially many solutions. We propose considering the (subset) maximal lgig guarded solutions as the \emph{canonical solutions} of the problem. Note that there are only polynomially many such solutions.

This definition is motivated by Theorem~\ref{thm:key} and the preceding lemmas, which show that every solution is either guarded by a lgig or can be specialised to a solution guarded by a lgig. In this sense, our canonical  solutions are  representative solutions of the problem. It is then straightforward to adapt Algorithm~\ref{alg} to instead enumerate all of the canonical solutions in polynomial time. 

\paragraph{Guarded rules in Datalog$^\pm$}
Guarded rules are particularly popular in Datalog$^\pm$, where they guarantee decidability for query answering. Recall, a rule in Datalog$^\pm$ can contain existential quantification in the head. For example, the following Datalog $^\pm$ rule can express that every person $x$ has an ancestor $y$.
\[
   \exists y Ancestor(y, x) \leftarrow Person(x) 
\]
Hence, the notion of guardedness of a rule is adapted such that the guard must only cover all variables that are not existentially quantified.

This is not an issue for our algorithm. The analogue to learning a guarded Datalog$^\pm$ rule in our setting would be to learn Horn clauses where only the negative literals must be guarded. Any unguarded variables in the positive literal correspond to existentially quantified variables. By adapting the notion of lgig map accordingly for positive atoms (positions with no shield are mapped to fresh variables) and distinguishing positive and negative literals, Algorithm~\ref{alg} can easily be adapted to a polynomial time algorithm for learning guarded Datalog$^\pm$ rules.

\subsection{Generalisation to $k$-Guardedness}

We have shown that the restriction to guarded hypotheses greatly
reduces the complexity of the ILP consistency problem. While guarded clauses have been repeatedly shown to exhibit a
rare balance between expressiveness and low complexity for reasoning,
some important concepts are not expressible in guarded form. One
particularly important such case is transitive closure, e.g., we can construct the transitive closure $T$ over a relation $R$ for example with the rule
\[
  T(x,y), R(y, z) \rightarrow T(x, z)
\]
but not in terms of a guarded rule.

However, the clause is clearly 2-guarded, as are many other common
logical properties and -- from anecdotal observations -- as are many
rules derived from expert knowledge in industry.  In this section we will briefly discuss how our
results from Sections~\ref{sec:det} and \ref{subsec:alg} can be
generalised to also hold for $k$-guarded clauses, where $k$ is
constant.

Overall, no major changes are necessary to accommodate
$k$-guardedness. The results of~\ref{sec:det} rely on the fact that
the variables in the guard have a uniquely-determined assignment in
any clause that is subsumed by the guarded hypothesis. This
observation can be extended to the case of a set of $k$ literals that
$k$-guard the hypothesis by viewing them as one \emph{merged}
literal $G$ whose argument list is simply the concatenation of the
individual argument lists. That is, in the transitive closure rule
above we could consider the merged guard $G(x,y,y,z)$ instead of the
literals $\{\neg T(x,y), neg R(y, z)\}$ that guard the respective
clause. It is straightforward to verify that such a merged guard behaves exactly the same as any normal guard literal.

There are $\binom{|C|}{k} \leq |C|^k$ possible ways to replace
$k$ literals in a clause $C$ with such a merged guard. Hence, an
adapted Algorithm~\ref{alg} for $k$-guarded hypothesis would have to try only
$|C|^k$ possible guard candidates. The computation of the lgig and
lgig maps is also still polynomial since the arity of the merged guard
is simply the sum of the individual literal arities. By
Lemma~\ref{prop:kguardsubsumption} the subsumption checks remain
tractable under $k$-guardedness. We therefore can also state the following more general tractability result.

\begin{theorem}
  Let $\mathcal{G}^*_k$ be the ILP setting like $\mathcal{G}^*$ where
  the hypothesis language is relaxed to allow $k$-guarded
  clauses. Then the ILP consistency problem for $\mathcal{G}_k^*$ is
  in \ptime.
\end{theorem}

\section{Conclusion \& Outlook}
\label{sec:conclusion}
We have studied the complexity of the ILP consistency problem when
learning guarded clauses.  In general, the problem is \np-complete, in
contrast to the case of learning Horn clauses with Horn examples,
which is \sigmatwo-complete. We then show how the restriction of
positive examples to \straight clauses allows for efficient learning
of guarded clauses, an important step for applications on large datasets. The tractability result also extends to the even more general class
$k$-guarded rules.

An interesting question that was left open is the integration of
background knowledge. Once we allow background knowledge with
variables we are confronted with conceptual issues on how whether and how the guardedness constraint should extend to background knowledge or not.

\bibliographystyle{acm}
\bibliography{refs}

\appendix
\setcounter{theorem}{12} %

\section{Full proofs for Section~\ref{sec:hardness}}

\begin{proof}[Proof of Lemma~\ref{lem:smallH}]
  Let $H$' be a solution for the instance $\expos, \exneg$ in the setting $\mathcal{G}$. We show that there exists a solution $H$ with $H \subseteq H'$ such that $|H| \leq 1+\sum_{E \in \exneg} |E|$. Note that any subset of $H$ clearly still subsumes $\expos$. We can therefore focus our construction on the non-subsumption of $\exneg$.

  We will construct $H$ from $H'$ iteratively along the negative examples $\{E_1, E_2, E_\ell\} \in \exneg$. To begin, let $G$ be a guard literal of $H'$, and let $H_0 = \{G\}$.
  Now, for $i$ from $1$ to $\ell$ in order, we consider the example $E_i$ and the question whether $H_{i-1} \not\vdash E_i$. If the subsumption does not hold, then we simply let $H_i = H_{i-1}$. In the case that $H_{i-1}$ does actually subsume $E_i$, we can make a key observation about guarded hypothesis.

  Since the guard contains all variables of $H_{i-1}$, any  variable substitution $\theta$  such that $H_{i-1}\theta \subseteq E_i$ corresponds to some mapping of $G$ into an occurrence of the signed relation symbol of the guard in $E_i$. Hence, there are at most $|E_i|$ such $\theta$. Let $\Theta_i$ be the respective set of substitutions such that $H_{i-1}\theta \subseteq E_i$ if and only if $\theta \in \Theta_i$. Since we know that $H' \not\vdash E_i$ and $G \in H'$, it follows that for every $\theta_j \in \Theta_i$, there is some literal $L_j \in H'$ such that $L_j\theta_j \not \in E_i$.

  Now let $H_i = H_{i-1} \cup \{ L_j \mid \theta_j \in \Theta_i\}$. From the argument above it follows that $H_i \not\vdash E_i$. Since for every possible variable substitution $\theta$, either $G\theta \not \in E_i$, or $\theta \in \Theta_i$ and thus also there is some $L_j \in H_i$ with $L_j \theta \not \in E_i$. 

  Let the final solution $H$ be the clause $H_\ell$ from the above construction. Since adding literals to the premise can never change non-subsumption to subsumption of a clause, we have $H \not\vdash \exneg$. As discussed above, from $H \subseteq H'$ it also follows that $H\vdash \expos$ and thus $H$ is a solution for the instance.

  What is left, is to show the size bound on $H$. Clearly, in every step from $H_{i-1}$ to $H_i$ for $1\leq i \leq \ell$, we have
  \[
    |H_i| \leq |H_{i-1}|+|\Theta_i| \leq |H_{i-1}|+|E_i|
  \]
  and therefore also $|H|=|H_k| \leq 1+\sum_{i=1}^\ell |E_i|$ (recall that $|H_0| = 1$).
\end{proof}

\section{Full proofs for Section~\ref{sec:alg}}

\begin{proof}[Proof of Lemma~\ref{lem:shielding}]
Suppose a atom $Q(x_1, \dots, x_n)$ in $H$ that violates the stated property. That is, there is some $i$ such that for all $j$ where $x_i=\pi^H_{P,j}$ we have $j \not \in S_{P,Q,i}$.

In consequence, for any $j$ such that $x_i = \pi^H_{P,j}$, there must be some clause $C\in E$ such that $j \not \in S^C_{P,Q,i}$. By definition then $\pi^C_{P,j} \neq \pi^C_{Q,i}$, i.e., the term at position $i$ in $Q$ does not match the term at position $j$ in $P$ in clause $C$. Since we assume that $x_i=\pi^H_{P,j}$ the same variable would need to map to $\pi^C_{P,j} $ as well as $\pi^C_{Q,i}$ to subsume the clause (recall, every relation symbol occurs only once). Hence $H$ does not subsume $C$ and we arrive a contradiction.
\end{proof}

The following two lemmas only serve to simplify part of the eventual proof of Theorem~\ref{thm:key}.

\begin{lemma}
  \label{thm:guarddet}
  Let $E$ be a set of \straight grounded clauses and $H$ a clause where the guard is the lgig of some relation $P$ induced by $E$, and $H \vdash E$.
  Then $H$ is a \straight clause.
\end{lemma}
\begin{proof}
  The statement follows from Lemma~\ref{lem:uniqshield}. Any two occurrences of the same signed relation symbol must, by the lemma, necessarily have the same variables at every position. In other words there is only one possible allowed variable list for every signed relation symbol in the hypothesis. Recall that we do not consider clauses to have duplicate literals (redundant by idempotence) and thus $H$ is \straight.
\end{proof}

\begin{lemma}
\label{lem:negex}
    Let $C$ be a clause and $H$ a hypothesis such that $H \not \vdash C$. Let $\theta$ be some specialisation of $\vars(H)$. It holds that $H\theta \not\vdash C$.
\end{lemma}
\begin{proof}
  Proof is by contradiction. Suppose $H\theta\vdash C$, then by definition there exists a substitution $\mu$ such that $H\theta\mu \subseteq C$. Thus also $H (\mu \circ \theta) \subseteq C$, contradicting our initial assumption that $H \not \vdash C$.
\end{proof}

\begin{proof}[Proof of Theorem~\ref{thm:key}]
Let $\theta$ be the specialisation as in Lemma~\ref{lem:specialize}. The clause $H\theta = \{C\theta \mid C \in H\}$ will be the $H'$ as stated in the theorem. The first property follows directly from Lemma~\ref{lem:specialize}. We proceed to demonstrate the other properties. 

Let us refer to the lgig (obtained from the guard of $H$ and the specialisation $\theta$) as $P(\bar{y})$ and recall that $y_i=y_j$ if and only if $\forall C \in \expos \colon \pi^C_{P,i} = \pi^C_{P,j}$.
Hence, for every clause $C \in \expos$ and every substitution $\mu$ such that $H\mu \subseteq C$, if positions $i$ and $j$ are the same in the lgig, then $\mu$ will map $\pi^H_{P,i}$ and $\pi^H_{P,j}$ to the exact same terms.
Since $H$ is guarded, this determines the substitution for all literals in $H$. One can then observe that we always also have $H\mu = H\theta\mu$ since the theta either maps $x \mapsto y$ in the case above, where $\mu$ maps $x$ and $y$ to the same terms, or $x \mapsto x$ otherwise.
Since we have $H'= H\theta$, we also have  $H'\mu \subseteq C$.

In other words, for every clause $C \in \expos$ we have $H' \vdash C$ and thus also $H' \vdash \expos$. By Lemma~\ref{lem:negex} we also have $H\theta \not \vdash \exneg$, i.e., $H' \not \vdash \exneg$.
Finally, the fact that  $H'$ is \straight follows directly from properties 1 and 3, in combination with Lemma~\ref{thm:guarddet}.
\end{proof}

\begin{proof}[Proof of Theorem~\ref{lem:algcomplete}]
  Soundness is straightforward as discussed abvoe. Correctness of the algorithm follows form the following completeness argument. In particular, we show that if a solution exists for the input instance, then Algorithm 1 will return a solution (i.e., not reject).
  
  Assume that the instance has a solution. By Theorem~\ref{thm:key} we
  know that there also exists a solution $H'$ that is guarded by an
  lgig $J$. A literal with the same signed relation symbol as $J$
  clearly exists in every clause in $\expos$ and thus, the outer loop
  of the algorithm (line 2) will find our target lgig in line 3. Either another solution is returned before the loop reaches $J$ as the value for $G'$, in which case the algorithm will return a solution, or we continue the algorithm with lgig $G' = J$. For the rest of the argument we thus assume $G'=J$.

  By Lemma~\ref{lem:uniqshield}, for any signed relation name that occurs in $H'$, a lgig map
  exists and is uniquely determined. Hence the check at line 8 will succeed for every signed relation in $H'$. Furthermore, if the signed relation name of $L$ occurs in $H'$, then always $\mu(L)\in H'$, i.e., the lgig map will always map the literals to the exact corresponding  literals in $H'$.
  We claim that in consequence, the set $H$ constructed by the algorithm, by the loop in lines 7-11, will always be a superset of $H'$ (modulo isomorphism). Thus, since $H' \not\vdash \exneg$ by assumption, we also have $H\not\vdash \exneg$ and the algorithm will return $H$.

  To argue the claim, we only need to observe that any addition of literals to $H$ in line 11 can never block adding the literals of $H'$ from being added. That is, for any $P \in H'$, and any $H$ guarded by $J$, if $H \vdash \expos$, then also $\{P\} \cup H \vdash \expos$ (and thus the check at line 10 will always succeed in adding literals from $H'$ to $H$).
  This can be seen by observing that there exists exactly one substitution $\theta$, determined by the guard, such that a $H\theta \subset C$ for any straight clause $C \in \expos$. Since $H$ and $H'$ share the same guard $J$, and both subsume $\expos$, it follows that this variable substitution is the same for both of them on each clause in $\expos$, i.e., $H \cup H' \vdash \expos$ and thus also $\{P\} \cup H \vdash \expos$.
\end{proof}

 \end{document}